\documentclass[11pt]{article}
\usepackage{amsmath,amsfonts,amssymb,amsthm}
\usepackage{fullpage}
\usepackage[ruled]{algorithm2e}

\usepackage{subfigure} \usepackage{epstopdf}
\usepackage{graphicx}
\usepackage{color}
\usepackage[sans]{dsfont}
\usepackage{mathtools}
\usepackage{hyperref}
\usepackage{cleveref}
\usepackage{tcolorbox}
\usepackage{multirow}
\usepackage[shortlabels]{enumitem}
\usepackage{xspace}
\usepackage[short]{optidef}
\tcbuselibrary{skins,breakable}
\tcbset{enhanced jigsaw}
\usepackage[font=small,labelfont=bf]{caption}
\usepackage{bbm}

\newtheorem{theorem}{Theorem}

\newtheorem{observation}[theorem]{Observation}

\theoremstyle{definition}
\newtheorem{claim}[theorem]{Claim}

\newcommand{\set}[1]{\left\{ #1 \right\}}

\newcommand{\cset}{{\mathcal{C}}}

\newcommand{\sset}{{\mathcal{S}}}

\newenvironment{properties}[2][0]
{
	\begin{enumerate} \setcounter{enumi}{#1}}{\end{enumerate}}

\newcommand{\dist}{\textnormal{\textsf{dist}}}

\newcommand{\ball}{\mathsf{B}}

\def\card#1{\left| #1 \right|}

\newcommand{\cut}{\mathsf{cut}}
\newcommand{\D}{\mathsf{D}}
\newcommand{\supp}{\mathsf{supp}}

\newcounter{note}

\begin{document}

	\begin{titlepage}
		
		\title{Towards the Characterization of Terminal Cut Functions:\\ a Condition for Laminar Families}

		\author{Yu Chen\thanks{EPFL, Lausanne, Switzerland. Email: {\tt yu.chen@epfl.ch}. Supported by ERC Starting Grant 759471.} \and Zihan Tan\thanks{Rutgers University, NJ, USA. Email: {\tt zihantan1993@gmail.com}. Supported by a grant to DIMACS from the Simons Foundation (820931).}} 
		
		\maketitle

		\thispagestyle{empty}
		\begin{abstract}
We study the following characterization problem. 
Given a set $T$ of terminals and a $(2^{|T|}-2)$-dimensional vector $\pi$ whose coordinates are indexed by proper subsets of $T$, is there a graph $G$ that contains $T$, such that for all subsets $\emptyset\subsetneq S\subsetneq T$, $\pi_S$ equals the value of the min-cut in $G$ separating $S$ from $T\setminus S$?
The only known necessary conditions are submodularity and a special class of linear inequalities given by Chaudhuri, Subrahmanyam, Wagner and Zaroliagis.

Our main result is a new class of linear inequalities concerning laminar families, that generalize all previous inequalities.
Using our new class of inequalities, we can prove a generalization of Karger's approximate min-cut counting result to graphs with terminals.		
		\end{abstract}
	\end{titlepage}

\section{Introduction}

The main idea behind vertex sparsification is that, although graphs are large, the vertices of interest to
us (referred to as \emph{terminals}) are potentially quite few. This enables the compression of the input
graph into a more compact form while preserving certain cut, flow, and distance information with respect to the terminals. In the past decades, much progress has been made in exploring the limit of vertex sparsification 
\cite{hagerup1998characterizing,moitra2009approximation,leighton2010extensions,charikar2010vertex,chuzhoy2012vertex,kratsch2012representative,krauthgamer2013mimicking,andoni2014towards,khan2014mimicking,krauthgamer2017refined,karpov2017exponential,goranci2017improved,chen2020fast,goranci2021expander,chang2022almost,krauthgamer2023exact,chen20231}. These results furthered our understanding of graph structures and also led to several interesting algorithmic implications.

In this work we study the following basic characterization problem on vertex sparsification for cut structures:
Given a set $T$ of terminals and a $(2^{|T|}-2)$-dimensional vector $\pi$ whose coordinates are indexed by proper subsets of $T$, is there a graph $G$ that contains $T$, such that for all subsets $S\subseteq T$, $\pi_S$ equals the value of the min-cut in $G$ separating $S$ from $T\setminus S$? (We call those vectors $\pi$ with a \textsf{YES} answer \emph{realizable}.)

Despite extensive work on vertex sparsification for cuts, our understanding on this basic problem remains limited. The only folklore result, which is a necessary condition here, is that \emph{terminal cut functions are submodular}.
Specifically, let $G$ be a graph and let $T$ be a set of its vertices designated as \emph{terminals}.
For each subset $\emptyset\subsetneq S\subsetneq T$, let $\cut_G(S)$ be the value of the min-cut in $G$ separating $S$ from $T\setminus S$.
Then for any $S,S'\subseteq T$, 
\[\cut_G(S)+\cut_G(S')\ge \cut_G(S\cap S')+\cut_G(S\cup S').\]
(Here $\cut_G(\cdot)$ is called the \emph{terminal cut function} of $G$. Note that the terminal cut function can also be viewed as a $(2^{|T|}-2)$-dimensional vector $(\cut_G(S))_{\emptyset \subsetneq S\subsetneq T}$.)


It was observed by Khan, Raghavendra \cite{khan2014mimicking} that the set of all realizable vectors  forms a \emph{convex cone}\footnote{Specifically, if graph $G$ realizes vector $\pi$ and graph $G'$ realizes vector $\pi'$, then for any pair of real numbers $a,b>0$, we can scale the edge capacities of $G,G'$ (by $a,b$, respectively) and then identify the corresponding terminals to obtain a graph that realizes vector $a\pi+b\pi'$.}. Therefore, from basic results in linear programming, we are indeed able to characterize the set of realizable vectors by a system of (possibly infinitely many) homogeneous linear inequalities.
For example, the submodularity condition can be written as: for all subsets $S,S'\subseteq T$, $\pi_S+\pi_{S'}\ge \pi_{S\cap S'}+\pi_{S\cup S'}$.
%
To formally define these inequalities, we say $\beta$ is a \emph{type vector}, iff $\beta$ is a $(2^{|T|}-2)$-dimensional nonnegative real vector with its coordinates indexed by the proper subsets of $T$. That is, $\beta=(\beta_S)_{\emptyset\subsetneq S\subsetneq T}$ where $\beta_S\ge 0$ for all $\emptyset\subsetneq S\subsetneq T$. 
Then a homogeneous linear inequality is simply $\langle \beta,\pi\rangle\ge \langle \gamma,\pi\rangle$, for a pair $\beta,\gamma$ of type vectors.
Now the question becomes: 
\[\emph{Which inequalities characterize all terminal cut functions, or equivalently,  all realizable vectors?}\]

Let us make some observations. First, for an inequality $\langle \beta,\pi\rangle\ge \langle \gamma,\pi\rangle$ to be valid (that is, satisfied by all terminal cut functions), any pair of terminals cannot be ``more separated'' in RHS than in LHS.
Formally, we say that a subset $S\subseteq T$ \emph{separates} a pair $t,t'$ of terminals, iff exactly one of $t,t'$ belongs to $S$. Now if for some pair $t,t'$, RHS has a term $\pi_S$ with $S$ separating $t,t'$, while LHS does not contain such a term, then the inequality cannot be correct. This is because, if we consider the graph $G$ on $T$ with only a single edge $(t,t')$, then the value of LHS is $0$ and the value of RHS is strictly positive.

To make this observation more general, for a type vector $\beta$, we define $\D_\beta$, the \emph{metric on $T$ induced by $\beta$}, as follows. For every pair $t,t'$ of terminals,
\[\D_\beta(t,t')=\sum_{S}\beta_S\cdot\mathbbm{1}[S \text{ separates } t,t'].\]
Note that $\D_\beta(\cdot,\cdot)$ is indeed a valid metric on $T$ as it is the positive combination of a collection of cut metrics on $T$.
Then this ``terminals no more separated in RHS than in LHS'' condition is essentially requiring that \emph{metric $\D_{\beta}$ dominates $\D_{\gamma}$}, that is, for every pair $t,t'\in T$, $\D_{\beta}(t,t') \ge \D_{\gamma}(t,t')$.


However, this condition alone is necessary but not sufficient. For example, when set $T=\set{1,2,3,4}$, type vector $\beta$ is such that $\beta_S=1/2$ for all doubleton subsets $S$ of $T$, and type vector $\gamma$ is such that $\gamma_S=1/2$ for all singleton subsets $S$ of $T$. Is one of $\langle \beta,\pi\rangle$ and $\langle \gamma,\pi\rangle$ always greater than the other?
Note that both $\D_{\beta}$ and $\D_{\gamma}$ are uniform metrics on $T$, so they dominate each other and we cannot refute either possibility, and their $\ge$/$\le$ relation cannot be deduced from submodularity, either.

Chaudhuri, Subrahmanyam, Wagner and Zaroliagis gave an answer by proving that $\langle \beta,\pi\rangle\ge \langle \gamma,\pi\rangle$ always holds in the above example.
Formally, they showed in \cite{ChaudhuriSWZ00} that, if $\D_\beta$ dominates $\D_\gamma$ and $\gamma$ only has non-zero values at coordinates indexed by singleton sets, then
$\langle \beta,\pi\rangle\ge \langle \gamma,\pi\rangle$ always holds.
These inequalities remain as the only known necessary constraints for terminal cut functions besides submodularity.


\subsection{Our Results}

Our main result is a new class of constraints for terminal cut functions, which generalizes both the result in \cite{ChaudhuriSWZ00} and the submodularity constraints \footnote{We discuss why they are indeed generalizations in \Cref{sec: discussion}.}.
For a type vector $\beta$, we define its \emph{support}, denoted by $\supp(\beta)$, as the collection of subsets $S$ with $\beta_S>0$. For any graph $G$, we define
\[\cut_G(\beta)=\sum_S\beta_S\cdot \cut_G(S).\]
Our main result is summarized as the following theorem.

\begin{theorem} \label{lem:cover}
Let $G$ be an edge-capacitated graph, $T$ be a set of its vertices, and $\beta,\gamma$ be type vectors. If (i) $\supp(\gamma)$ is a laminar family; and (ii) metric $\D_{\beta}$ dominates metric $\D_{\gamma}$ (that is, for every pair $t,t'\in T$, $\D_{\beta}(t,t') \ge \D_{\gamma}(t,t')$); 
then $\cut_G(\beta)\ge \cut_G(\gamma)$.
\end{theorem} 

Using our main result above, we can prove a generalization of Karger's approximate min-cut counting result. The initial result in \cite{karger1993global} stated that, in an $n$-vertex graph, the number of cuts whose value is at most $\alpha$ times the global min-cut is at most $2^{2\alpha}\cdot\binom{n}{2\alpha}$, for every half-integer $\alpha>0$. Later, in a subsequent work \cite{karger2000minimum}, it was pointed out that the number of $\alpha$-approximate global min-cut for an $n$-vertex cycle graph is at least
$$C_{n,\alpha}=\binom{n}{2}+\binom{n}{4}+\ldots+\binom{n}{2\alpha},$$
for every integer $\alpha>0$, and he then showed an almost matching upper bound of $(1+o(1))\cdot C_{n,\alpha}$ for all graphs, improving upon his previous bound $2^{2\alpha}\cdot\binom{n}{2\alpha}$.

We prove the following generalization of Karger's result to graphs with terminals.

\begin{theorem} \label{thm:terminal-karger}
	Let $G$ be a connected graph and $T$ be a set of its vertices. Then for any integer $\alpha\ge 1$, the number of subsets $S$ of $T$ with
	$\cut_G(S)\le \alpha\cdot \min_{\emptyset\subsetneq S'\subsetneq T}\set{\cut_G(S')}$ is at most $(1+o(1))\cdot C_{|T|,\alpha}$.
\end{theorem}

After the first version of the paper, we are notified that the same bound can be proved via Mader's edge-splitting technique and the result in \cite{karger2000minimum}. We sketch the proof in \Cref{sec: discussion}.
Techniquewise, our proof utilizes LP and Duality, and is fundamentally different from the randomized contraction algorithm which gave the bound in \cite{karger1993global}, or the approach in \cite{karger2000minimum} which used the tree-packing result of Tutte and Nash-Williams. 
We believe this duality-based approach using the laminar-family characterization of terminal cut functions should be of independent interest and have the potential of proving/improving more theorems.

\paragraph{Related Work.}
Beideman, Chandrasekaran, and Wang \cite{beideman2023approximate} proved an asymptotically same bound of $n^{O(\alpha)}$ for $\alpha$-approximate global min-cuts via a different approach, and showed an deterministic poly-time algorithms (for constant $\alpha$) for enumerating all of them. For smaller values of $\alpha$, Benczur \cite{benczur1995representation} showed an $O(n^2)$ upper bound for the number of $6/5$-approximate global min-cut. Later, Nagamochi, Nishimura and Ibaraki \cite{nagamochi1997computing} showed the same bound for $4/3$-approximate global min-cuts, and then Henzinger and Williamson \cite{henzinger1996number} showed the same bound for the number of cuts that are strictly smaller than $3/2$ times the global min-cut.
There have also been some results on characterizing set functions that can be realized as cut functions (but not terminal cut functions) of graphs (both directed and undirected, but without terminals) and special families of hypergraphs \cite{fujishige2001realization,yamaguchi2016realizing}.

\paragraph{Organization.}
We provide the proof of \Cref{lem:cover} in \Cref{sec: main} and the proof of \Cref{thm:terminal-karger} in \Cref{sec: Karger}. Finally in \Cref{sec: discussion}, we give some discussions on our results and future works.

\section{Proof of \Cref{lem:cover}}
\label{sec: main}

In this section, we prove \Cref{lem:cover}.
We denote by $G=(V,E,c)$ the input graph, where $\set{c(e)}_{e\in E}$ are edge capacities. The value of a cut is the total capacity of all edges in the cut, so all minimum cuts discussed in this section are minimum-capacity cuts.

At a high level, we ``combine'' fractional min-cuts (which can be interpreted as edge length functions in the LP for computing min-cuts) separating $S$ from $T\setminus S$ for all subsets $S\in \supp(\beta)$, and then ``extract'' from it fractional min-cuts separating $S'$ from $T\setminus S'$ for all $S'\in \supp(\gamma)$, thereby proving that the total min-cut size for $S'\in \supp(\gamma)$ is at most the total min-cut size for $S\in \supp(\beta)$.


We start by assigning weights (lengths), which are irrelevant to capacities, to edges of $G$, as follows. 
Initially, all edges $e\in E$ have length $0$.
For each set $S\in \supp(\beta)$, we compute a minimum cut in $G$ separating $S$ from $T\setminus S$, and increase the length of each edge in this cut by $\beta_S$. We denote by $\ell_e$ the length of each edge $e$ after we have processed all sets $S\in \supp(\beta)$, and denote by $\dist_\ell(\cdot,\cdot)$ the shortest-path distance metric on $V$ induced by edge lengths $\set{\ell_e}_{e\in E}$. We start by proving the following observations.

\begin{observation}
\label{obs: dist}
For every pair $t,t'\in T$, $\dist_{\ell}(t,t')\ge \D_{\beta}(t,t')$.
\end{observation}
\begin{proof}
Consider any path $P$ connecting $t,t'$ in $G$.
Let $S$ be a set in $\supp(\beta)$ that cuts $t,t'$. From our algorithm, we computed a min-cut separating $S$ from $T\setminus S$ and increased the length of its edges by $\beta_S$, so this cut also separates $t$ from $t'$, and therefore must contain an edge in $P$, whose length was increased by $\beta_S$ in this iteration. Therefore, the length of $P$ is at least $\sum_{S: S\text{ cuts }t,t'}\beta_S$, which by definition is $\D_\beta(t,t')$.
\end{proof}


\begin{observation}
\label{obs: cut size}
$\cut_G(\beta)=\sum_{e}c(e)\cdot \ell_e$.
\end{observation}
\begin{proof}
At the beginning, every edge has length $\ell_e=0$ and so $\sum_{e}c(e)\cdot \ell_e=0$. In the iteration of processing $S$, we computed a min-cut $E'$ in $G$ separating $S$ from $T\setminus S$, with $\sum_{e\in E'}c(e)=\cut_G(S)$, and increased the length of each of its edges by $\beta_S$, so the total increase in the sum $\sum_{e}c(e)\cdot \ell_e$ is $\sum_{e\in E'}c(e)\cdot \beta_S=\beta_S\cdot \cut_G(S)$. Therefore, in the end, $\sum_{e}c(e)\cdot \ell_e=\sum_S\beta_S\cdot \cut_G(S)=\cut_G(\beta)$. 
\end{proof}

We will now construct another collection $\set{\ell^S}_{S\in \supp(\gamma)}$ of edge length functions, such that
\begin{properties}{P}
\item for each edge $e\in E(G)$, $\ell_e\ge \sum_{S\in \supp(\gamma)}\ell^S_e$; and
\label{prop: length bound}
\item for each set $S\in \supp(\gamma)$, $\sum_e c(e)\cdot \ell^S_e\ge \gamma_S\cdot \cut_G(S)$.
\label{prop: cut size}
\end{properties}
Note that, if such edge length functions exist, then from \Cref{obs: cut size},
\[
\begin{split}
\cut_G(\beta) & =\sum_{e}c(e)\cdot \ell_e\ge \sum_{e}c(e)\cdot \bigg(\sum_{S\in \supp(\gamma)}\ell^S_e\bigg)\\
& = \sum_{S\in \supp(\gamma)}\bigg(\sum_e c(e)\cdot\ell^S_e\bigg)\ge \sum_{S\in \supp(\gamma)}\gamma_S\cdot \cut_G(S)=\cut_G(\gamma),
\end{split}
\]
and we are done. Therefore, from now on we focus on constructing such edge length functions.

\newcommand{\con}{\textnormal{\textsf{con}}}
\newcommand{\lcon}{\ell^{\con}}
\newcommand{\Vcon}{V^{\con}}

\subsubsection*{Step 1. Continuization of a graph}

We introduce the notion of \emph{continuization} of a graph, where each edge is viewed as a continuous line segment full of points, and a metric on these points are naturally induced by the edge lengths.

Let $G=(V,E,\ell)$ be an edge-weighted graph. Its \emph{continuization} is a metric space $(V^{\con},\ell^{\con})$, that is defined as follows. 
Each edge $(u,v)\in E$ is viewed as a continuous line segment $\con(u,v)$ of length $\ell_{(u,v)}$ connecting $u,v$, and the point set $V^{\con}$ is the union of the points in all lines $\set{\con(u,v)}_{(u,v)\in E}$. 
Specifically, for each edge $(u,v)\in E$, the line $\con(u,v)$ is defined as \[\con(u,v)=\set{(u,\delta)\mid 0\le \delta\le \ell_{(u,v)}}=\set{(v,\theta)\mid 0\le \theta\le \ell_{(u,v)}},\]
where $(u,\delta)$ refers to the unique point in the line that is at distance $\delta$ from $u$, and $(v,\theta)$ refers to the unique point in the line that is at distance $\theta$ from $v$, so $(u,\delta)=(v,\ell_{(u,v)}-\delta)$.

The metric $\ell^{\con}$ on $V^{\con}$ is essentially the geodesic distance induced by the shortest-path distance metric $\dist_{\ell}(\cdot,\cdot)$ on $V$.
Formally, for a pair $p,p'$ of points in $V^{\con}$, 
\begin{itemize}
	\item if $p,p'$ lie in the same line $(u,v)$, say $p=(u,\delta)$ and $p'=(u,\delta')$, then $\lcon(p,p')=|\delta-\delta'|$;
	\item if $p$ lies in the line $(u,v)$ with $p=(u,\delta)$ and $p'$ lies in the line $(u',v')$ with $p'=(u',\delta')$, then 
	\[
	\begin{split}
	\lcon(p,p')=\min\{
	&  \dist_{\ell}(u,u')+\delta+\delta', \quad
	\dist_{\ell}(u,v')+\delta+(\ell_{(u',v')}-\delta'),\\
	&	\dist_{\ell}(v,u')+(\ell_{(u,v)}-\delta)+\delta', \quad
	\dist_{\ell}(v,v')+(\ell_{(u,v)}-\delta)+(\ell_{(u',v')}-\delta')
	\}.
	\end{split}
	\]
\end{itemize}
Clearly, every vertex $u\in V$ also belongs to $\Vcon$, and for a pair $u,u'\in V$, $\dist_{\ell}(u,u')=\lcon(u,u')$.
See \Cref{fig: con} for an illustration.

\begin{figure}[h]
	\centering
	\scalebox{0.15}{\includegraphics{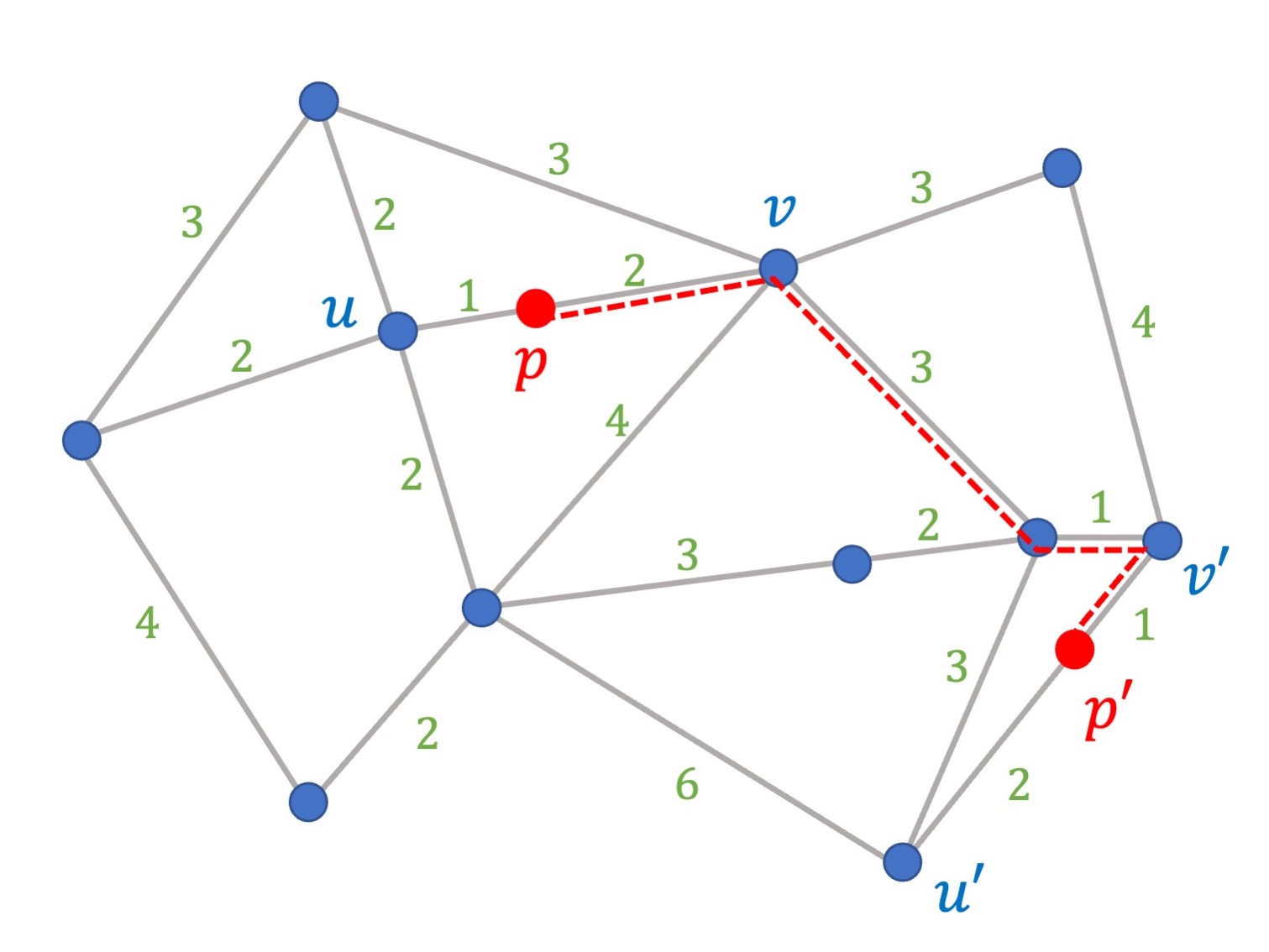}}
	\caption{An illustration of the continuization of a graph. Line $\con(u,v)$ has length $3$ and a point $p$ in this line is $p=(u,1)=(v,2)$. Line $\con(u',v')$ has length $3$ and a point $p'$ in this line is $p'=(u',2)=(v',1)$. The shortest (in $\lcon$) path connecting $p$ to $p'$ in $\Vcon$ is shown in the red dashed line.\label{fig: con}}
\end{figure}

For a path $P$ in $G$ connecting $u$ to $u'$, it naturally corresponds to a set $P^{\con}$ of points in $\Vcon$, which is the union of all lines corresponding to edges in $E(P)$. The set $P^{\con}$ naturally inherits the metric $\lcon$ restricted on $P^{\con}$. We will also call $P^{\con}$ a path in the continuization $(\Vcon,\lcon)$.
A subset of a line segment $\con(u,v)$ of the form $\set{(u,\delta)\mid a\le \delta< b}$ (where $0\le a<b\le \ell_{(u,v)}$) is called a \emph{sub-segment}, and its \emph{length} is defined as $(b-a)$.

\subsubsection*{Step 2. Balls, regions and length functions}

For a vertex $u\in V$ and a real number $r\ge 0$, we define the \emph{ball around $u$ with radius $r$}, denoted by $\ball(u,r)$, as the collection of all points in $\Vcon$ that is at distance (in $\lcon$) less than $r$ from $u$. That is,
\[
\ball(u,r)=\bigg\{p\in \Vcon \text{ }\bigg| \text{ } 0\le \lcon(p,u)< r\bigg\}.
\]

We now use balls to define regions and then construct edge length functions. 
We iteratively process sets $S$ in $\supp(\gamma)$ and define its region $\Phi_S$, as follows. 
Throughout, we maintain a collection $\set{r_t}_{t\in T}$ of radii for all terminals. Initially, set $r_t=0$ for all $t\in T$.
In each iteration, we pick a set $S$ such that $S$ is not processed but all other sets $S'$ in $\supp(\gamma)$ with $S'\subseteq S$ have been processed. We then define the \emph{region of $S$} as (see \Cref{fig: ball} for an illustration)
\[\Phi_S=\bigg(\bigcup_{t\in S}\ball(t,r_t+\gamma_S)\bigg)\setminus \bigg(\bigcup_{t\in S}\ball(t,r_t)\bigg).\]

\begin{figure}[h]
	\centering
	\scalebox{0.15}{\includegraphics{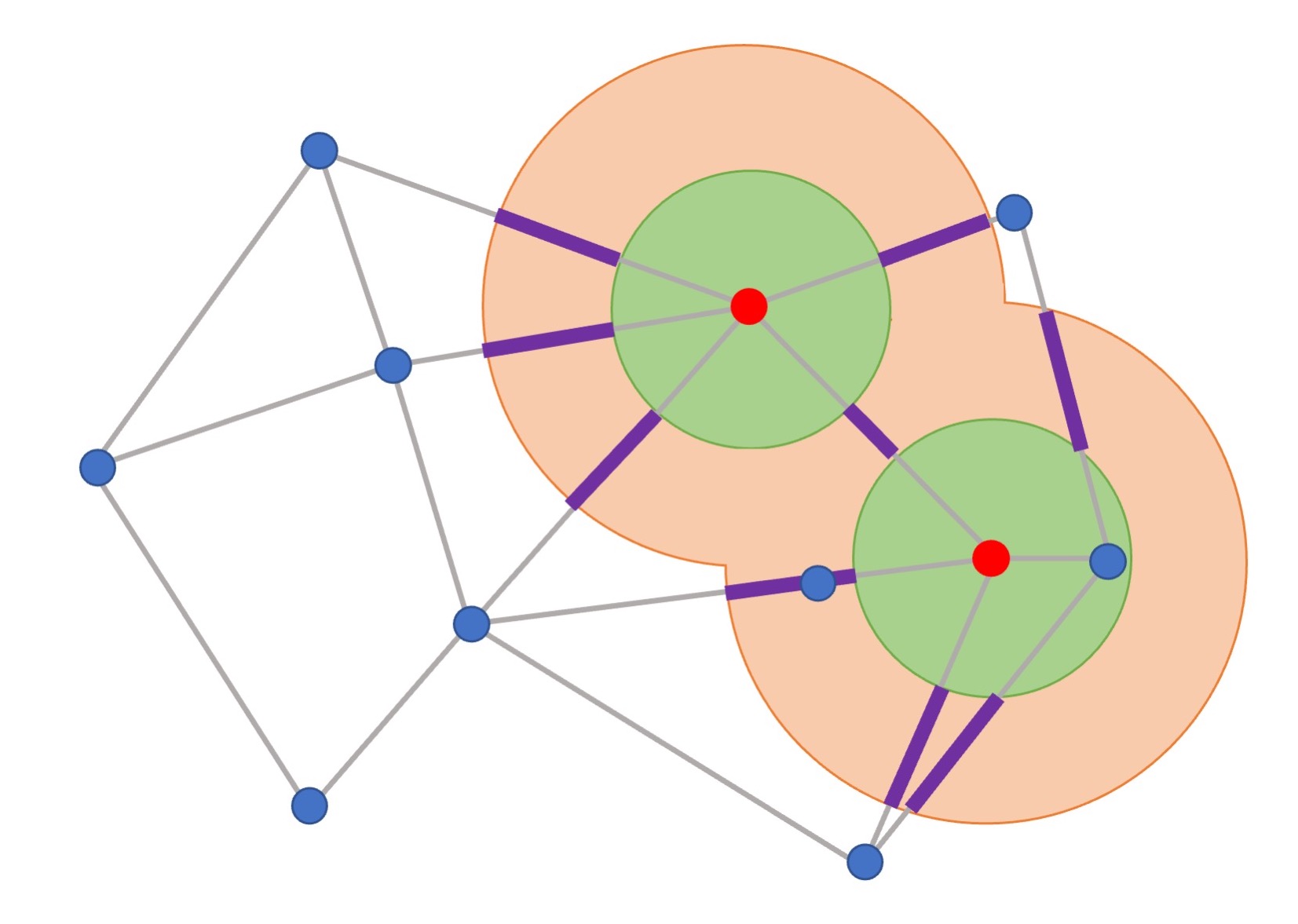}}
	\caption{An illustration of balls and regions. The terminals in $S$ are shown in red. The balls $\set{\ball(t,r_t)}_{t\in S}$ are shown in green. The region $\Phi_S$ is the union of all sub-segments shown in purple.\label{fig: ball}}
\end{figure}

We then increase $r_t$ by $\gamma_S$ for all $t\in S$, and continue to the next iteration. 
We keep processing sets in $\supp(\gamma)$ in this way, until all regions $\set{\Phi_S}_{S\in \supp(\gamma)}$ have been constructed.

From our algorithm of growing balls and the definition of metric $\lcon(\cdot,\cdot)$, it is immediate to observe that, for each line $\con(u,v)$ in $\Vcon$ and each region $\Phi_S$, their intersection is either 
%
\begin{properties}{C}
\item an empty set; or
\label{case: empty}
\item a sub-segment of $\con(u,v)$; or 
\label{case: 1}
\item the union of two sub-segments of $\con(u,v)$.
\label{case: 2}
\end{properties}

In each case, we set length $\ell^S_e$ of the edge $e=(u,v)$ as the \emph{total length} of the intersection. That is, in Case \ref{case: empty}, we set $\ell^S_e=0$; in Case \ref{case: 1}, we set $\ell^S_e$ as the length of the sub-segment; and in Case \ref{case: 2}, we set $\ell^S_e$ as the total length of two sub-segments. This completes the definition of $\ell^S$.

It remains to show that Properties \ref{prop: length bound} and \ref{prop: cut size} hold, which we prove in the following claims.

\begin{claim}
\label{clm: length}
For each edge $e\in E(G)$, $\ell_e\ge \sum_{S\in \supp(\gamma)}\ell^S_e$.
\end{claim}
\begin{proof}
We only need to show that, for every pair $S,S'\in \supp(\gamma)$, their regions $\Phi_S,\Phi_{S'}$ do not intersect. Note that this implies that, for each edge $e\in E(G)$, the intersections $\set{\con(e)\cap \Phi_S}_{S\in \supp(\gamma)}$ are mutually disjoint (sub-segments of $\con(e)$). As each length $\ell^S_e$ is the total length of its corresponding intersection $\con(e)\cap \Phi_S$, the total length $\sum_{S}\ell^S_e$ is no more than $\ell_e$.

Recall that $\supp(\gamma)$ is a laminar family, so either $S\subseteq S'$, or $S'\subseteq S$, or $S\cap S'=\emptyset$.

We first assume that $S\subseteq S'$ (the case where $S'\subseteq S$ is symmetric). From the algorithm, we have processed $S$ before $S'$. If the radii after processing $S'$ are $\set{r'_t}_{t\in S'}$, then for each $t\in S$, the radius $r_t$ after processing $S$ must be strictly smaller than $r'_t$, as according to the algorithm we have increased this radius by $\gamma_{S'}>0$ in the iteration of processing $S'$. Therefore, $\Phi_S\subseteq \big(\bigcup_{t\in S}\ball(t,r_t)\big)$, and so it is disjoint from $\Phi_{S'}$ by definition of $\Phi_{S'}$.

Assume now that $S\cap S'=\emptyset$, and without loss of generality that $S'$ is processed after $S$. Let $\set{r_t}_{t\in T}$ be the radii after processing $S'$. By definition, $\Phi_S\subseteq \big(\bigcup_{t\in S}\ball(t,r_t)\big)$ and $\Phi_{S'}\subseteq \big(\bigcup_{t'\in S'}\ball(t',r_{t'})\big)$. It remains to show that, for each pair $t\in S,t'\in S'$, $\ball(t,r_t)\cap \ball(t',r_{t'})=\emptyset$.
Consider such a pair $t\in S, t'\in S'$.
Note that, until this iteration, we have never processed any set $S''$ containing both $t$ and $t'$, as otherwise $S'$ must be a subset of $S''$ (as $\supp(\gamma)$ is a laminar family and $S''\setminus S'\ne \emptyset$) and should have been processed before set $S''$, a contradiction.
Thus, in each of the previous iterations, at most one of $r_t$ and $r_{t'}$ was increased, and that happened only when we processed a set $S''$ that cuts $t,t'$, which increased either $r_t$ or $r_{t'}$ by $\gamma_{S''}$. Therefore, up to this iteration,
\[r_t+r_{t'}\le \sum_{S''\in \supp(\gamma)}\gamma_{S''}\cdot\mathbbm{1}[S'' \text{ cuts } t,t']=\D_{\gamma}(t,t')\le \D_{\beta}(t,t')\le \dist_{\ell}(t,t'),\] 
where the last inequality follows from \Cref{obs: dist}. As a result, $\ball(t,r_t)\cap \ball(t',r_{t'})=\emptyset$ holds for every pair $t\in S, t'\in S'$, and it follows that $\Phi_S\cap \Phi_{S'}=\emptyset$. 
\end{proof}

\begin{claim}
For each set $S\in \supp(\gamma)$, $\sum_e c(e)\cdot \ell^S_e\ge \gamma_S\cdot \cut_G(S)$.
\end{claim}
\begin{proof}
We consider the following LP, which is the dual of the standard (path based) max-flow LP, and is used for computing the value $\cut_G(S)$ of the minimum cut in $G$ separating $S$ from $T\setminus S$.
\begin{eqnarray*}
	\mbox{(LP-mincut)}\quad	\quad & \text{minimize} &\sum_{e\in E(G)}c(e)\cdot x_e \\
	&\sum_{e\in E(P)}x_e\geq 1 &\forall \text{ path }P\text{ connecting a vertex of }S \text{ to a vertex of }T\setminus S\\
	&x_e\geq 0&\forall e\in E(G)
\end{eqnarray*}
The first constraints can be interpreted as: in the shortest-path distance metric induced by $G$ with edges weighted by $\set{x_e}$, the distance between any vertex of $S$ and any vertex of $T\setminus S$ is at least $1$.
Therefore, in order to prove the claim, it suffices to show that, in the shortest-path distance metric induced by $G$ with edges weighted by $\set{\ell^S_e}$, the distance between any vertex of $S$ and any vertex of $T\setminus S$ is at least $\gamma_S$.

Consider any path $P$ in $G$ (with edges weighted by $\set{\ell_e}$) connecting a terminal $t\in S$ to a terminal in $T\setminus S$. Consider the corresponding path $P^{\con}$ in the continuation $\Vcon$ as directed from $t$ to its endpoint in $T\setminus S$.
Consider the function $f(p)=\lcon(p,t)$ defined on all points $p$ in $P^{\con}$. 
It is easy to see that, as the point $p$ moves along $P^{\con}$, the value $\ell(p,t)$ evolves continuously, either increasing or decreasing at the same rate with $p$. Therefore, $f$ is a piece-wise linear function with each piece having slope $1$ or $-1$.

We denote by $\set{r_t}_{t\in T}$ the radii of terminals before processing set $S$.
Consider the last point on $P^{\con}$ that lies in $\bigcup_{t\in S}\ball(t,r_t)$, which we denote by $p'$. 
Assume that $p'$ is at distance $r_{t'}$ from terminal $t'$ (note that it is possible that $t'\ne t$). Consider the points in the set
\[\bigg\{p\text{ }\bigg| \text{ } p\text{ lies on the subpath of $P^{\con}$ between $p'$ and $t'$, and }r_{t'}\le \ell(p,t')\le r_{t'}+\gamma_S\bigg\}.\] 
On the one hand, from the definition of point $p'$, this set is completely disjoint from $\bigcup_{t\in S}\ball(t,r_t)$.
On the other hand, from similar arguments in the proof of \Cref{clm: length} (by analyzing the distance in $\dist_{\ell}$ between $t'$ and the other endpoint of $P^{\con}$ in $T\setminus S$), we can show the set indeed contains some point $p$ with $\lcon(p,t')=r_{t'}+\gamma_S$.
Therefore, from the properties of $f$, the points in this set form sub-segments in $\Vcon$ whose total length is at least $\gamma_S$, and it follows from the definition of $\ell_S$ that $\sum_{e\in E(P)}\ell^S_e\ge \gamma_S$. 
As $P$ is an arbitrary path connecting a terminal in $S$ to a terminals in $T\setminus S$, the distance (in $\dist_{\ell^S}$) between any vertex of $S$ and any vertex of $T\setminus S$ is at least $\gamma_S$.
This completes the proof of the claim.
\end{proof}

\section{Proof of \Cref{thm:terminal-karger}}
\label{sec: Karger}

In this section, we provide the proof of \Cref{thm:terminal-karger}. 
Recall that we are given a graph $G$ and a set $T$ of its vertices.
Denote $|T|=k$.
We define $\Pi= \min_{S'\ne \emptyset,T}\set{\cut_G(S')}$, so $\Pi$ is the minimum terminal cut value in $G$.
Specifically, we will show that the number of subsets $S$ of $T$ with $\cut_G(S)\le \alpha\cdot \Pi$ is at most
\[
K_{k,\alpha}=\sum_{i=1}^{2\alpha}\binom{k-1}{i}\cdot(2\alpha+1-i).
\]
Note that
\[
C_{k,\alpha}=\sum_{i=1}^{\alpha}\binom{k}{2i}=\sum_{i=1}^{\alpha}\bigg(\binom{k-1}{2i}+\binom{k-1}{2i-1}\bigg)=\sum_{i=1}^{2\alpha}\binom{k-1}{i}.
\]
It is easy to verify that $K_{k,\alpha}=(1+o(1))\cdot C_{k,\alpha}$. From now on we denote $K=K_{k,\alpha}$ for convenience.


Assume for contradiction that there exist more than $K$ such subsets $S$. We define type vector $\beta$ as the indicator vector for these subsets, i.e., $\beta_S=\mathbbm{1}[\cut_G(S)\le \alpha\cdot \Pi]$, so  $\cut_G(\beta)\le \alpha\cdot K\cdot \Pi$.
We will show that there exists another type vector $\gamma$, such that
\begin{properties}{I}
\item $\supp(\gamma)$ is a laminar family;
\label{lam}
\item for every pair $t,t'\in T$, $\D_{\beta}(t,t') \ge \D_{\gamma}(t,t')$; and
\label{dom}
\item $\sum_{S}\gamma_S > \alpha \cdot K$.
\label{cut}
\end{properties}
Note that \ref{lam} and \ref{dom} imply that $\beta$ and $\gamma$ satisfy the conditions in \Cref{lem:cover}, so $\cut_G(\beta)\ge \cut_G(\gamma)$, and then from \ref{cut}, $\cut_G(\beta)\ge \cut_G(\gamma) =\sum_{S}\gamma_S\cdot \cut_G(S)\ge \sum_{S}\gamma_S\cdot \Pi> \alpha \cdot  K \cdot \Pi$, a contradiction.

From now on we focus on the existence of such a type vector $\gamma$. We consider the following LP.
%
%
\begin{eqnarray*}
	\mbox{(LP-Primal)}\quad	\quad & \text{maximize} &\sum_{S}\gamma_S \\
	&\sum_{S \text{ cuts } t,t'}\gamma_S\leq \D_{\beta}(t,t') &\forall t,t'\in T\\
	&\gamma_S\geq 0&\forall S\subseteq T, S\ne \emptyset,T
\end{eqnarray*}

The existence of vector $\gamma$ satisfying conditions \ref{lam}, \ref{dom}, and \ref{cut} is guaranteed by the following claims.

\begin{claim} \label{obs:uncross}
    There is an optimal solution $\gamma^*$ to (LP-Primal) such that $\supp(\gamma^*)$ is a laminar family.
\end{claim}
\begin{proof}
Let $\gamma^*$ be the optimal solution to (LP-Primal) that, among all optimal solutions, minimizes $\sum_{S}\gamma_S\cdot |S|$.
Consider the collection $\supp(\gamma^*)$, and assume for contradiction that it is not a laminar family. Then there exist sets $S,S'$ such that 
\begin{itemize}
\item $\gamma^*_S,\gamma^*_{S'}> 0$ (assume w.l.o.g that $\gamma^*_S\le \gamma^*_{S'}$); and
\item $S\cap S',S\setminus S',S'\setminus S$ are all non-empty.	
\end{itemize}
We construct another vector $\hat\gamma$ as follows: \[\hat\gamma_S = 0, \hat\gamma_{S'} = \gamma^*_{S'}-\gamma^*_S, \hat\gamma_{S\setminus S'} = \gamma^*_{S\setminus S'}+\gamma^*_S, \hat\gamma_{S'\setminus S} = \gamma^*_{S'\setminus S}+\gamma^*_S, \text{ and for other sets $\hat S$ }, \hat\gamma_{\hat S}=\gamma_{\hat S}.\]	
In other words, we decrease the coordinates indexed by $S$ and $S'$ by $\gamma^*_S$, and increase the coordinates indexed by $S\setminus S'$, and $S'\setminus S$ by $\gamma^*_S$. 

On the one hand, for each pair $t,t'\in T$, it is easy to verify that
\[\mathbbm{1}[S\setminus S' \text{ cuts }t,t']+\mathbbm{1}[S'\setminus S \text{ cuts }t,t']\le \mathbbm{1}[S \text{ cuts }t,t']+\mathbbm{1}[S' \text{ cuts }t,t'],\]
so $\hat\gamma$ is also a feasible solution to (LP-Primal).

On the other hand, it is immediate to verify that
$\sum_{S}\hat\gamma_S\cdot |S|=\big(\sum_{S} \gamma^*_S\cdot |S|\big)-2\cdot \gamma^*_S\cdot |S\cap S'|$. As $\gamma^*_S, |S\cap S'|>0$, this is a contradiction to the assumption that $\gamma^*$ minimizes $\sum_{S}\gamma_S\cdot |S|$ among all optimal solutions to (LP-Primal).
\end{proof}


\begin{claim}
The optimal value of (LP-Primal) is greater than $\alpha \cdot K$.
\end{claim}
\begin{proof}
Consider the following LP, which is the dual of (LP-Primal).
\begin{eqnarray*}
	\mbox{(LP-Dual)}\quad	\quad & \text{minimize} &\sum_{t,t'\in T}x_{t,t'}\cdot \D_{\beta}(t,t') \\
	&\sum_{t\in S,t'\notin S}x_{t,t'}\geq 1 &\forall S\subseteq T, S\ne \emptyset,T\\
	&x_{t,t'}\geq 0&\forall t,t'\in T
\end{eqnarray*}

Note that (LP-Dual) is essentially the LP relaxation of the metric minimum spanning tree problem for metric $\D_{\beta}(\cdot,\cdot)$. Specifically, we are given a complete graph $H$ on $T$ with edges $(t,t')$ equipped with weight $\D_{\beta}(t,t')$, and the goal is to find a minimum spanning tree in $H$. In fact, it follows from \cite{jain2001factor} (Theorem 3.1) that the integrality gap of (LP-Dual) is at most $2$.
Therefore, it suffices to show that the MST cost of $H$ is greater than $2\alpha K$.

Let $\tau$ be any spanning tree of $H$. 
Recall that $\beta$ is the indicator vector for all subsets $S\subseteq T$ with $\cut_G(S)\le \alpha\cdot \Pi$.
By definition, the cost of $\tau$ is
\begin{equation}
\label{eqn}
\begin{split}
\sum_{(t,t')\in E(\tau)} \D_{\beta}(t,t') & = \sum_{(t,t')\in E(\tau)} \sum_{S}\beta_S\cdot\mathbbm{1}[S \text{ cuts } t,t']\\
& = \sum_{S\in \supp(\beta)} \sum_{(t,t')\in E(\tau)}\mathbbm{1}[S \text{ cuts } t,t']\\
& = \sum_{S\in \supp(\beta)} \big(\#\text{ of edges in $\tau$ cut by }S\big).
\end{split}
\end{equation}
On the other hand, we prove the following observation.
\begin{observation}
For any subset $E'\subseteq E(\tau)$, there exist exactly two sets $S\subseteq T$ that cuts $\tau$ at $E'$.
\end{observation}
\begin{proof}
Let $\tau_1,\ldots,\tau_r$ be the subtrees obtained from $\tau$ by removing all edges of $E'$. Consider a set $S\subseteq T$ that cuts $\tau$ at $E'$. By definition of cutting an edge, for each $1\le i\le r$, set $S$ includes either all vertices in $\tau_i$ or none of them. Moreover, if $E'$ contains an edge between $\tau_i$ and $\tau_j$ (in this case we say $\tau_i$ and $\tau_j$ are \emph{neighbors}), then exactly one of $\tau_i$ and $\tau_j$ is included in $S$. Therefore, such a set $S$ either (i) includes $\tau_1$, excludes all neighbors of $\tau_1$, includes all neighbors of neighbors of $\tau_1$, ...; or (ii) excludes $\tau_1$, includes all neighbors of $\tau_1$, excludes all neighbors of neighbors of $\tau_1$, ...;
and it follows that there exist exactly two sets $S\subseteq T$ that cuts $\tau$ at $E'$.
\end{proof}

As $\tau$ is a tree on $k$ vertices, $|E(\tau)|=k-1$.
So for each $1\le i\le k$, the number of $i$-edge subset of $E(\tau)$ is $\binom{k-1}{i}$, and therefore the number of sets $S\subseteq T$ that cuts $\tau$ at exactly $i$ edges is $2\cdot \binom{k-1}{i}$. Thus, (recall that we have assumed that $|\supp(\beta)|>K$)
\[
\begin{split}
\sum_{S\in \supp(\beta)} \big(\#\text{ of edges in $\tau$ cut by }S\big)  & = \sum_{i\ge 1} i\cdot \big(\#\text{ of $S$ in $\supp(\beta)$ that cuts $\tau$ at exactly $i$ edges}\big)   \\
& \ge (2\alpha+1)\cdot \sum_{i\ge 2\alpha+1} \big(\#\text{ of $S$ in $\supp(\beta)$ that cuts $\tau$ at exactly $i$ edges}\big) \\
& \quad\quad+  \sum_{i=1}^{2\alpha} i\cdot \big(\#\text{ of $S$ in $\supp(\beta)$ that cuts $\tau$ at exactly $i$ edges}\big)\\
& \ge (2\alpha+1) K - 2\cdot\bigg(\binom{k-1}{2\alpha}+2\cdot \binom{k-1}{2\alpha-1}+\cdots + 2\alpha\cdot \binom{k-1}{1}\bigg) \\
& =  2\alpha K.
\end{split}
\]
This, together with \ref{eqn}, implies that $\sum_{(t,t')\in E(\tau)} \D_{\beta}(t,t')>2\alpha K$. It follows that the optimal value of (LP-Dual) is greater than $\alpha K$. Finally, from Strong Duality, the optimal value of (LP-Primal) is greater than $\alpha K$.
\end{proof}

\section{Discussion and Future Work}
\label{sec: discussion}

We first sketch an alternative proof of \Cref{thm:terminal-karger}, that was pointed to us by anonymous reviewers.
 
\emph{Proof Sketch of \Cref{thm:terminal-karger} using Mader's Edge Splitting Technique.}

Mader has proved in \cite{mader1978reduction} the following theorem.

\begin{theorem}
Let $G$ be an undirected unweighted graph where each vertex has an even degree. For every vertex $u\in V(G)$, there exists a pair $v,v'$ of neighbors of $u$, such that in the graph obtained from $G$ by replacing edges $(u,v),(u,v')$ with edge $(v,v')$, the $x$-$y$ edge connectivity stays the same as in $G$ for every pair $x,y\in V(G)\setminus \set{u}$.
\end{theorem}

This operation of replacing edges $(u,v),(u,v')$ with edge $(v,v')$ is called \emph{edge-splitting}.
As a corollary of this theorem, we can repeatedly apply edge-splitting to obtain a graph $H$ that only contains terminals and preserve pairwise edge-connectivity between every pair of terminals (we can assume without loss of generality that every vertex in $G$ has even degree since we can duplicate each edge and subdivide them). Moreover, in the resulting graph $H$,
\begin{itemize}
\item the global min-cut in $H$ equals the global terminal-separating min-cut in $G$, as for any pair $t,t'$ of terminals lying on different side of the global terminal-separating min-cut in $G$, their edge-connectivity is preserved in $H$; and
\item any $\alpha$-approximate terminal-separating min-cut in $G$ is an $\alpha$-approximate global min-cut in $H$, as the min-cut separating any terminal partition $(T_1,T_2)$ of $T$ only gets decreased in the edge-splitting process.
\end{itemize}
Therefore, the number of $\alpha$-approximate terminal-separating min-cut in $G$ is at most the number of $\alpha$-approximate global min-cut in $H$, which, by the result of \cite{karger2000minimum}, is bounded by $(1+o(1))\cdot C_{|T|,\alpha}$.

$\ $

We now provide some discussion and comparison between our result and previous results.

First, as singleton subsets of $T$ are mutually disjoint, they form a laminar family, and so our main result (\Cref{lem:cover}) generalizes the previous result of \cite{ChaudhuriSWZ00}.
On the other hand, consider the following inequality implied by \Cref{lem:cover}: (where $T=\set{1,2,3,4,5,6}$, and for simplicity all $\cut_G$ notations are omitted)
\[
\begin{split}
& \set{1,2,3}+\set{1,2,4}+\set{1,2,5}+\set{1,2,6}+\set{1,3,4}+\set{1,5,6} \\
& \ge\set{1}+\set{2}+\set{1,2}+\set{3}+\set{4}+\set{3,4}
+\set{5}+\set{6}+\set{5,6}.
\end{split}
\]
This inequality cannot be deduced by any combination of inequalities proved in \cite{ChaudhuriSWZ00}, so our result strictly generalizes their result.

Second, consider a submodular constraint $\cut_G(S)+\cut_G(S')\ge \cut_G(S\cap S')+\cut_G(S\cup S')$. It is easy to verify that 
for each pair $t,t'\in T$,
\[\mathbbm{1}[S\setminus S' \text{ cuts }t,t']+\mathbbm{1}[S'\setminus S \text{ cuts }t,t']\le \mathbbm{1}[S \text{ cuts }t,t']+\mathbbm{1}[S' \text{ cuts }t,t'].\]
Observe that $S\cap S' \subseteq S\cup S'$, so RHS is a laminar family. Therefore, our main result (\Cref{lem:cover}) also generalizes all the submodular constraints.
On the other hand, the following inequality (implied by \Cref{lem:cover}) cannot be deduced from submodular constraints (where $T=\set{1,2,3,4}$):
\[\cut_G(\set{1,2})+\cut_G(\set{1,3})+\cut_G(\set{1,4})\ge \cut_G(\set{1})+\cut_G(\set{2})+\cut_G(\set{3})+\cut_G(\set{4}),\]
since all submodular constraints satisfy that $\sum_S \beta_S=\sum_S\gamma_S$, while the above inequality does not. Therefore, our main result strictly generalizes the submodularity.

The most interesting open problem is to find the characterization of all terminal cut functions by a system of homogeneous linear constraints. Two special cases are potentially good starting points. 
\begin{itemize}
\item Find the characterization for the $|T|=6$ case ($|T|\le 5$ was settled by \cite{ChaudhuriSWZ00}).
\item Find all ``symmetric'' constraints: a type vector $\beta$ is \emph{symmetric}, iff for each $1\le s\le |T|$, the values $\beta_S$ are the same for all sets $S$ with $|S|=s$. For symmetric $\beta$ and $\gamma$, which inequalities $\langle \beta,\pi\rangle\ge \langle \gamma,\pi\rangle$ are correct/characterizing?
\end{itemize}

\paragraph{Acknowledgement.} We would like to thank anonymous reviewers for pointing the alternative proof in \Cref{sec: discussion} to us, and for the helpful comments in improving the presentation of this paper.

\bibliographystyle{alpha}
\bibliography{REF}

\end{document}